\documentclass[a4paper]{article}

\usepackage[english]{babel}
\usepackage[utf8x]{inputenc}
\usepackage[T1]{fontenc}

\usepackage[a4paper,top=3cm,bottom=2cm,left=3cm,right=3cm,marginparwidth=1.75cm]{geometry}

\usepackage{amsmath}
\usepackage{graphicx}
\usepackage[colorinlistoftodos]{todonotes}
\usepackage[colorlinks=true, allcolors=blue]{hyperref}
\usepackage{amssymb}
\usepackage{amsthm}
\usepackage{comment}
\usepackage{hyperref}
\usepackage{textcomp}
\usepackage{listings}

\theoremstyle{plain}
\newtheorem{theorem}{Theorem}[section]

\newtheorem{lemma}[theorem]{Lemma}

\theoremstyle{definition}
\newtheorem{definition}[theorem]{Definition}

\theoremstyle{remark}

\newtheorem*{theorem*}{Theorem}

\newcommand{\rr}{\raggedright}
\newcommand{\tn}{\tabularnewline}

\definecolor{clGray}{rgb}{0.5,0.5,0.5}
\definecolor{clNote}{rgb}{08,0,0}

\newcommand{\ONESTM}{one--state Turing machine}
\newcommand{\ONESTMS}{one--state Turing machines}

\newcommand{\MCC}{$M_{cc}$}
\newcommand{\LANGCC}{$L_{cc}$}

\title{The Power of One-State Turing Machines}
\author{Marzio De Biasi}
\date{}

\begin{document}
\maketitle

% ================================================
% ==== Abstract
% ================================================

\begin{abstract}
	At first glance, one--state Turing machines are very weak: the halting problem for them is decidable, and, without memory, they cannot even accept a simple one element language such as $L = \{ 1 \}$ . Nevertheless it has been showed that a one--state Turing machine can accept non regular languages. We extend such result and prove that they can also recognize non context--free languages, so for some tasks they are more powerful than pushdown automata.

%\keywords
%    one--state Turing machines, formal languages
\end{abstract}

% ================================================
% ==== Introduction
% ================================================
\section{Introduction}

	A Turing machine with two states is able to simulate any Turing machine; hence we can build a two--states universal Turing machine (actually 2 states and 18 symbols are enough to achieve universality \cite{DBLP:journals/fuin/NearyW0}) hence the halting problem for them is undecidable. We dramatically clips their wings if we allow only one state: we get a device with no internal memory and the halting problem becomes decidable \cite{HermanHalt,saouter:inria-00074105}.  We cannot even distinguish between accepting and non--accepting states. Nevertheless we can relax the notion of acceptance and rejection of an input: a one--state Turing machine accepts a string $x$ if it halts; it rejects $x$ if it runs forever. Under this different notion of acceptance, a one--state Turing machine is more powerful than finite--state automata \cite{HERMAN1969353}. In \cite{KUDLEK1996241} Kudlek considers languages accepted by Turing machines with one state and three symbols, and with two states and two symbols. He proves that one state Turing machines with 3 symbols accept regular languages, except in one case a deterministic linear context-free language.
    
    We extend such result and prove that a one--state Turing machine is more powerful than a pushdown automata, too: it can recognize non context--free languages.
    
    In Section 2 we give some preliminary definitions, in Section 3 we build a \ONESTM{} that is able to \emph{compare} two numbers in different bases, in Section 4 we give a formal proof that the language recognized by such machine is not context--free.
    
       %Finally we give some considerations about the computational complexity of deciding if a one--state Turing machine accepts or rejects a given input $x$

% ================================================
% ==== Definitions
% ================================================
\section{Definitions}

	We use the standard definition of Turing machine as a 7-tuple \cite{Sipser}:

\begin{definition}[Turing Machine]
\label{def:tm}
A Turing machine $M$ is a 7-tuple
	$\langle Q, \Sigma, \Gamma, q_0, q_{a}, q_{r}, \delta \rangle$ where
	$Q$ is the set of states, 
    $\Sigma$ is the input alphabet not containing the \emph{blank symbol} $\sqcup$,
    $\Gamma$ is the tape alphabet, $\sqcup \in \Gamma, \Sigma \subseteq \Gamma$,
	$q_0 \in Q$ is the initial state,
	$q_{a} \in Q$ is the accept halting state,
	$q_{r} \in Q$ is the reject halting state ($q_r \neq q_a$),
	$\delta : Q \setminus \{q_a, q_r\} \times \Gamma \to Q \times \Gamma \times \{L,R\}$
	is the transition function that given a non-halting state and the current symbol returns a new (possibly halting) state, a symbol that is written on the tape and the direction of the move ($Left$ or $Right$).
\end{definition}

The single \emph{head} of the Turing machine is placed on a \emph{tape} with an infinite number of \emph{cells} in both directions; at the beginning of the computation the state of the Turing machine is $q_0$, the head is placed on the leftmost symbol of the input string; the rest of the cells contain the blank symbol $\sqcup$.

	With two states and 18 symbols ($|Q| = 2, |\Gamma|=18$) we can build an universal Turing machine \cite{DBLP:journals/fuin/NearyW0}; we restrict our study to \ONESTMS{} so $|Q| = 1$ \cite{HERMAN1969353,saouter:inria-00074105}.
But with only one state available we cannot neither distinguish between an accepting and rejecting state nor between a non--halting and a halting state; so we must somewhat alter the standard definition of accept/reject:

\begin{definition}[One--state Turing machine]
A one state Turing machine is a 6-tuple:\\
	$\langle \{q\}, \Sigma, \Gamma, q, H, \delta \rangle$ where
    $q$ is the only state which is also the initial state,
    $\Sigma$ is the input alphabet not containing the \emph{blank symbol} $\sqcup$,
    $\Gamma$ is the tape alphabet, $\sqcup \in \Gamma, \Sigma \subseteq \Gamma$,
    $H$ is the set of \emph{halting symbols},
	$\delta : \Gamma \setminus H \to \Gamma \times \{L,R\}$
    is the transition function that given the current symbol under the head, returns the new symbol to be written and the direction of the move. If the current symbol is a halting symbol, the machine halts.
\end{definition}

We define the language recognized by a \ONESTM{}: 

\begin{definition}[Language recognized by a \ONESTM{}]
\label{def:langrec}
If $M$ is a \ONESTM{}, $x \in \Sigma^*$, we define the language recognized by $M$ the set: 
$$L(M) = \{ x \mid M \text{ halts on input x} \}$$.
\end{definition}

	We could also consider the variant in which the \ONESTM{} is equipped with  distinct  sets of \emph{accept} and \emph{reject} halting symbols. In this case the behavior can be: halt and accept, halt and reject or run forever. However it is straightforward to see that the results would also extend to this model.

	At first glance, a \ONESTM{} is very weak, indeed it cannot even recognize trivial regular languages:

\begin{theorem}
\label{thm:lowpow}
No \ONESTM{} $M$ can recognize the trivial regular language $L = \{ 1 \}$.
\end{theorem}

\begin{proof}
Suppose that $M$ recognizes $L = \{1\}$; we analyze its behavior on input $x=1$. At the beginning of the computation the head is on the symbol $1$, if it is a halting symbol then $M$ halts, but it will also halt and accept on any string $x = 1^+$. If it is not a halting symbol, either $\delta(1) = (a_1,L)$ \emph{(i)} or $\delta(1) = (a_1,R)$ \emph{(ii)}, $a_1 \in \Gamma$.
\begin{itemize}

\item[(i)] Suppose that $\delta(1) = (a_1,L)$: now the head is on a blank symbol; if it is halting, then again it will also halt and accept on any $x = 1^+$; it cannot move left again otherwise it will never halt; so it must move back to the right, i.e. $\delta(\sqcup)=(a_{\sqcup}, R), a_{\sqcup} \in \Gamma$. At this point the head is on symbol $a_1$ (written on the first step); if $\delta(a_1) = (a_1',R)$ then the head will end on a blank symbol on the right part of the tape and it will start an infinite sequence of right moves. So $a_1$ must be halting or move the head to the left: $a_1 \in H \lor \delta(a_1) = (a_1', L)$; in either cases $M$ would also halt on any string $x = 1^+$.

\item[(ii)] If $\delta(1) = (a_1,R)$ the behavior is similar: the head falls on a blank symbol and it must halt or jump back. If it halts, then $M$ would halt on any $x = 1^+$ after a sequence of right moves. If it jumps back, it cannot move left again on $a_1$ otherwise it would fall on a blank symbol on the left part of the input moving forever. So we must have $\delta(\sqcup) = (a_{\sqcup},L)$ and $a_1 \in H \lor \delta(a_1) = (a_1', R)$ ; but such transitions cause $M$ to halt and accept also all $x = 1^+$.
\end{itemize}

In all cases $L(M) \setminus \{1\} \neq \emptyset$ leading to a contradiction.
\end{proof}

Nevertheless, as showed in \cite{HERMAN1969353}, there are \ONESTMS{} that can recognize non--regular languages, i.e. some of them are more powerful than finite--state automata. 

In the next sections we will prove that they are even more powerful.

% ================================================
% ==== A powerful one--state Turing machine
% ================================================
\section{A powerful one--state Turing machine}

	A \ONESTM{} has no internal memory, but the symbols on the tape cells can be used to \emph{store} the following information:

\begin{itemize}
	\item cell has been visited $n \bmod k$ times ($k$ fixed);
	\item head has left the cell going to the right;
	\item head has left the cell going to the left.
\end{itemize}

	This information is limited, but it is enough to implement both an \emph{unary counter} and a \emph{binary counter}. Informally we define a Turing machine with this behavior:

\begin{itemize}
	\item the input is divided in two parts: the left part ($U^*$) is an unary counter; whenever the head finds a $U$ it marks it as \emph{counted} writing $C$ and bounces back to the right part;
    
    \item the right part ($0^*$) -- initially filled with $0$s -- is a binary counter: whenever the head finds a zero $0$ it transforms it to a $1$ and bounces back to the left part; when it finds a $1$ it marks it as a zero $Z$ and continues towards the right to ``propagate'' the \emph{carry}; when going back to the left part it will change $Z$ to $0$;
    
    \item at the (rightmost) end of the input there is the special symbol $h$ for which no transition is defined and causes the machine to halt.
    
\end{itemize}

Formally, we define a \ONESTM{} $M_{cc} = \langle \{q\}, \Sigma, \Gamma, q, \{h\}, \delta \rangle$ over a three symbols input alphabet
$\Sigma = \{ u, 0, h \}$;
the tape alphabet has eight symbols:
$\Gamma = \{ \sqcup, u, U, h, 0, 1, Z, C, B \}$;
the complete transition table $\delta : \Gamma \to \Gamma \times \{L,R\} $ is:

\begin{center}
\begin{tabular}{ c || c | c | p{6cm} } 
   Read & Write & Move & Description\\
  \hline
   $u$  &  $U$  & R    & \rr Move towards the center and prepare the unary counter \tn
   $0$  &  $1$  & L    & \rr Increment binary bit and go back left\tn
   $U$  &  $C$  & R    & \rr Mark as counted an element of the unary counter \tn
   $1$  &  $Z$  & R    & \rr Increment binary bit and propagate the carry \tn
   $Z$  &  $0$  & L    & \rr Restore the zero and continue to the left \tn
   $C$  &  $B$  & L    & \rr Skip $C$ and continue to the left \tn
   $B$  &  $C$  & R    & \rr Restore the $C$s and continue to the right \tn
   $\sqcup$ & $\sqcup$ & L & \rr Run left forever \tn
   $h$  &    &   & \rr Undefined (halt) \tn
   
\end{tabular}
\end{center}

Figure~\ref{fig:comp} shows an example of a computation on input $uuuu00h$.

\begin{figure}[h]
\centering
\begin{tabular}{c}
\begin{lstlisting}
 _ [u] u  u  u  0  0  h  _   : starting configuration
 _  U [u] u  u  0  0  h  _   : move to the center
 _  U  U [u] u  0  0  h  _ 
 _  U  U  U [u] 0  0  h  _   : bin counter is 00
 _  U  U  U  U [0] 0  h  _ 
 _  U  U  U [U] 1  0  h  _   : bin counter is 10
 _  U  U  U  C [1] 0  h  _   : unary counter is C
 _  U  U  U  C  Z [0] h  _ 
 _  U  U  U  C [Z] 1  h  _ 
 _  U  U  U [C] 0  1  h  _   : bin counter is 01
 _  U  U [U] B  0  1  h  _ 
 _  U  U  C [B] 0  1  h  _ 
 _  U  U  C  C [0] 1  h  _   : unary counter is CC
 _  U  U  C [C] 1  1  h  _   : bin counter is 11
 _  U  U [C] B  1  1  h  _ 
 _  U [U] B  B  1  1  h  _ 
 _  U  C [B] B  1  1  h  _ 
 _  U  C  C [B] 1  1  h  _ 
 _  U  C  C  C [1] 1  h  _   : unary counter is CCC
 _  U  C  C  C  Z [1] h  _ 
 _  U  C  C  C  Z  Z [h] _   : halt
\end{lstlisting}
\end{tabular}
\caption{Computation on input $uuuu00h$ (``\_'' represents the blank symbol $\sqcup$)}
\label{fig:comp}
\end{figure}

We call \LANGCC{} the language defined by the inputs on which the Turing machine \MCC{} halts (i.e. reach the halting symbol $h$ for which no transition is defined and doesn't run forever):

$$L_{cc}= \{ x \mid M_{cc}(x) \text{ halts}\}$$

We don't need to characterize exactly the language $L_{cc}$; but if we restrict the input to strings of the form $u^* 0^* h$, we can prove the following property:

\begin{theorem}
\label{thm:bin}
If $x \in \{ u^n 0^m h \mid n,m \geq 0 \}$ then $x \in L_{cc}$ if and only if $n \geq 2^{m}-1$.
\end{theorem}

\begin{proof}
By construction the left part $u^n$ behaves like an unary counter; while the right part $0^m$ behaves like a binary counter. For every increment of the binary counter the unary counter is also incremented, so the head can reach the rightmost $h$ only if $n \geq 2^{m}-1$; if $n < 2^{m}-1$ then the head reaches the blank symbol on the left of the input and continues in that direction forever.
\end{proof}

We can say that, if the input is ``well-formed'', the \ONESTM{} can \emph{compare} a number represented in unary  with a number represented in binary.

Using a similar technique we could build a \ONESTM{} that implements an arbitrary $k$-ary counter or a \ONESTM{} that compare a $k$-ary number with a $k'$-ary number, i.e. it can compare two numbers written in different bases.

% ================================================
% ==== Computational Complexity
% ================================================
\section{Computational complexity}

We define the languages:

$$L_{u0h}  = \{ u^n 0^m h \mid n,m \geq 0 \}$$

$$L'_{cc} = \{ u^n 0^m h \mid n,m \geq 0 \land n \geq 2^{m}-1 \}$$

which, by Theorem~\ref{thm:bin}, are related to each other and to $L_{cc}$ in this way:

$$L'_{cc} \subseteq L_{u0h}$$
$$L'_{cc} \subseteq L_{cc} \cap L_{u0h}$$

We can prove that $L_{cc}$ is not context--free (CF) using the well known property of context--free languages: if $L$ is CF and $R$ is regular then $L\cap R$ is CF. 

The following is immediate:

\begin{lemma}
\label{lem:reg}
	$L_{u0h}= \{ u^n 0^m h \mid n,m \geq 0 \}$ is regular.
\end{lemma}

and using the pumping lemma for context--free languages we prove the following:

\begin{lemma} 
\label{lem:notcf}
   $L'_{cc}$ is not context--free.
\end{lemma}

\begin{proof}
  Suppose that $L'_{cc}$ is CF. Let $p \geq 1$ be the pumping length.
  We pick the string 

$$s = u^{2^p-1} 0^{p} h$$

which is in $L'_{cc}$.
By the pumping lemma $s$ can be written as $s = rvwxy$, with substrings $r, v, w, x$ and $y$, such that
\emph{i)} $|vwx| \leq p$, \emph{ii)} $|vx| \geq 1$, and \emph{iii)} $rv^n w x^n y \in  L'_{cc}$ for all $n \geq 0$. We can have the following cases:

\begin{itemize}

\item $vwx$ is entirely contained in $u^{2^p-1}$; in this case we can pump zero times ($n=0$) and we get the 
string $s' = u^k 0^p h$ with $k < 2^p - 1$, hence $s' \notin L'_{cc}$;

\item $vwx$ is entirely contained in $0^p$; in this case we can pump two times ($n = 2$) and we get the string $s' = u^{2^p-1} 0^j h$ with $p < j$ and $2^p-1 < 2^j-1$, hence $s' \notin L'_{cc}$;

\item $vwx$ spans between $u^{2^p-1}$ and $0^p$; in the worst case $v = u^{p-1}$ and $x=0^1$: 

$$ u^{2^p - p} (u^{p-1})^n (0^1)^n 0^{p-1}h \in L'_{cc}$$

should hold for any $n \geq 0$, but for large enough $n$ we have:

$$2^p - p + n (p-1) < 2^{n  + p - 1} - 1$$

which puts the string out of the language.
So we can conclude that $L'_{cc}$ is not context--free.
\end{itemize}
\end{proof}

We can finally conclude that the language recognized by the \ONESTM{} $M_{cc}$ is not context--free:

\begin{theorem}
\label{thm:onestmnotcf}
$L_{cc}$ is not context--free.
\end{theorem}

\begin{proof}
Suppose that $L_{cc}$ is CF. Then $L_{cc} \cap L_{u0h}$ is CF because $L_{u0h}$ is regular (Lemma~\ref{lem:reg}).
But $L_{cc} \cap L_{u0h} = L'_{cc}$ which is not CF (Lemma~\ref{lem:notcf}) leading to a contradiction.
\end{proof}

If $R$ is the class of regular languages, $CF$ is the class of Context--Free languages, and $TM_{1state}$ is the class of languages recognized by \ONESTMS{}; by Theorem~\ref{thm:lowpow} and Theorem~\ref{thm:onestmnotcf}, we have :

\begin{theorem}
$$R \not\subset TM_{1state},\quad\quad TM_{1state} \not\subset CF$$
\end{theorem}

% ================================================
% ==== Conclusion
% ================================================
\section{Conclusion}

	We proved that despite its limitations a \ONESTM{} can be used to implement a $k$-ary counter and ``compare'' two numbers written in different bases (the quotes are due to the non--standard definition of acceptance/reject of an input string). So the class of languages recognized by \ONESTMS{} is neither a superset of the regular languages nor a subset of the context-free languages.
	
	As a further investigation, we state the following problem: what is the minimal number $n$ of symbols needed by a one--state Turing machine to accept a non context-free language? It is known that $4 \le n \le 8$: $4 \le n$ proved in \cite{KUDLEK1996241} and $n \le 8$ proved in this article.
    
% Further investigations ??? Running time??? Nondeterminism ???

\begin{comment}
Fischer, P.C. (1965). On formalisms for Turing machines.
\end{comment}

\bibliographystyle{abbrv}
\bibliography{onestm}

\end{document}